\documentclass[11pt]{article}
\usepackage{amsmath}     % new environments
\usepackage{amsthm}
\usepackage{amssymb}
\usepackage[ruled,vlined]{algorithm2e}
\usepackage{mathrsfs}
 \usepackage{psfrag} 			%FIGURES!!!
\usepackage{algorithmic}	%ALGORITHM AM
\usepackage{a4wide}
\usepackage{graphicx}

\newtheorem{thm}{Theorem}

\newtheorem{cor}{Corollary}

\newtheorem{obs}{Observation}

\newenvironment{keyword}{\par{\noindent\bf Keywords:}}

\begin{document}

%*******TITLE AND AUTHORS*******************************************
\title{Combinatorial optimization problems with uncertain costs and the OWA criterion}

\author{Adam Kasperski\thanks{Corresponding author}\\
   {\small \textit{Institute of Industrial}}
  {\small \textit{Engineering and Management,}}
  {\small \textit{Wroc{\l}aw University of Technology,}}\\
  {\small \textit{Wybrze{\.z}e Wyspia{\'n}skiego 27,}}
  {\small \textit{50-370 Wroc{\l}aw, Poland,}}
  {\small \textit{adam.kasperski@pwr.edu.pl}}
  \and
  Pawe{\l} Zieli{\'n}ski\\
    {\small \textit{Institute of Mathematics}}
  {\small \textit{and Computer Science}}
  {\small \textit{Wroc{\l}aw University of Technology,}}\\
  {\small \textit{Wybrze{\.z}e Wyspia{\'n}skiego 27,}}
  {\small \textit{50-370 Wroc{\l}aw, Poland,}}
  {\small \textit{pawel.zielinski@pwr.edu.pl}}} 
  
  \date{}
    
\maketitle

\begin{abstract}
In this paper a class of combinatorial optimization problems with uncertain costs is discussed. The uncertainty is modeled by specifying a discrete scenario set containing $K$ distinct cost scenarios. The Ordered Weighted Averaging (OWA for short) aggregation operator is applied to choose a solution. Some well known criteria used in decision making under uncertainty such as the maximum, minimum, average, Hurwicz and median are special cases of OWA. Furthermore, by using OWA, the traditional robust (min-max) approach to combinatorial optimization problems with uncertain costs can be generalized.  The computational complexity and  approximability of the problem of minimizing OWA for the considered class of problems are investigated and some new positive and negative results in this area are provided. These results remain valid for many basic problems, such as network or resource allocation problems.\end{abstract}


\section{Introduction}

In many combinatorial optimization problems we seek an object composed of some elements of a finite set whose total cost is minimum. This is the case, for example, in an important class of network problems where the set of elements consists of all arcs of some network and we wish to find an object in this network such as a path, a spanning tree, or a matching whose total cost is minimum. In general, the combinatorial optimization problems can often be expressed as  0-1 programming problems with a linear objective function, where a binary variable is associated with each element and a set of constraints describes the set of feasible solutions. For a comprehensive description of this class of problems we refer the reader to~\cite{AH93, GN72, PS82}.

The usual assumption in combinatorial optimization is that all the element costs are precisely known. However, the assumption that all the costs are known in advance
 is often unrealistic. In practice,  before solving a problem, we only know a set of possible realizations of the element costs. 
 This set is called a \emph{scenario set} and each particular realization of  the element costs within this scenario set is called a \emph{scenario}.
  Several methods of defining scenario sets have been proposed in the existing literature. 
   The \emph{discrete} and \emph{interval} uncertainty representations are
   among the most popular 
   (see, e.g.,~\cite{KY97}). In the former, scenario set contains a finite number of explicitly given cost vectors. In the latter one, for each element an interval of its possible values is specified and scenario set is the Cartesian product of these intervals. In the discrete uncertainty representation, each scenario can model some event that has a global 
   influence on the element costs. On the other hand, the interval uncertainty representation is appropriate when each element cost may vary within some range independently on the values of the other costs. A modification of the interval uncertainty representation was proposed in~\cite{BS03}, where the authors assumed that only a fixed and a priori given number of costs may vary.
  More general scenario sets which can be used in mathematical programming problems were discussed, for example, in~\cite{BN09}.
In this paper we assume that no additional information, for example a probability distribution, for scenario set is provided.

If scenario set contains more than one scenario, then an additional criterion is required to choose a solution.
 In \emph{robust optimization} (see, e.g., \cite{BN09, KY97}) we typically seek a solution minimizing the worst case behavior over all scenarios. Hence the min-max and min-max regret criteria are widely applied. However, this approach to decision making is often regarded as too conservative or pessimistic (see, e.g., \cite{LR57}). In particular, the min-max criterion takes into account only one, the worst-case scenario, ignoring the information connected with the remaining scenarios. This criterion also assumes that decision makers are very risk averse, which is not always true.

In this paper 
we wish to investigate
a class of combinatorial optimization problems with the discrete uncertainty representation. Hence, a scenario set provided with the input data, contains a finite number of explicitly given cost scenarios. In order to choose a solution we propose to use the \emph{Ordered Weighted Averaging} aggregation operator (OWA for short) introduced by Yager in~\cite{YA88}. The OWA operator is widely applied to aggregate the criteria in multiobjective decision problems (see, e.g.,~\cite{GS12, OS03, YKB11}), but it can also be applied to choose a solution under the discrete uncertainty representation. It is enough to treat the cost of a given solution under $j$th scenario as a $j$th criterion. 
The key elements of the OWA operator are weights whose number equals the number of scenarios. The $j$th weight expresses an importance of the $j$th largest cost of a given solution. 
Hence, the weights allow a decision maker to take his attitude towards a risk into account and use  the information about all scenarios while computing a solution. The OWA operator generalizes the traditional criteria used in decision making under uncertainty such as the maximum, minimum, average, median, or Hurwicz criterion. So, by using  OWA  we can generalize the  min-max approach, typically used in robust optimization. Let us also point out that the OWA operator is a special case of \emph{Choquet integral},  a sophisticated tool for aggregating criteria in multiobjective decision problems  (see, e.g., \cite{GL08}). The Choquet integral has been recently applied to some multicriteria network problems in~\cite{GPS10}.

Unfortunately, the min-max combinatorial optimization problems are almost always harder to solve than their deterministic counterparts, even when the number of scenarios equals~2. In particular, 
 the min-max versions of the shortest path, minimum spanning tree, minimum assignment, minimum $s-t$ cut, and minimum selecting items problems are NP-hard even for~2 scenarios~\cite{ABV08, AV01,  KY97}. Furthermore, if the number of scenarios is a part of the input, then all these problems become strongly NP-hard and hard to approximate within any constant factor~\cite{KZ13, KZ09, KZ11}. Since the maximum criterion is a special case of OWA, the general problem of minimizing OWA is not easier. However, it is not difficult to show that some other particular cases of OWA, such as the minimum or average, lead to problems whose complexity is the same as the complexity of 
 their deterministic counterparts.
 It is therefore of interest to 
  provide a characterization of the problem complexity depending on various weight distributions.

In this paper we provide the following new results. In Section~\ref{sec2scen}, we study the case when the number of scenarios equals~2. We give a characterization of the problem complexity depending on the weight distribution. In Section~\ref{secfptas}, we show some sufficient conditions for the problem to admit a fully polynomial time approximation scheme (FPTAS), when the number of scenarios is constant. Finally, in Section~\ref{secunb}, we consider the case in which the number of scenarios is a part of the input. We discuss different types of weight distributions. We show that for nonincreasing weights (i.e. when larger weights are assigned to larger solution costs) and 
for the Hurwicz criterion, the problem admits an approximation algorithm whose worst case ratio depends on the problem parameters, in particular on the number of scenarios. On the other hand, we show that if the weights are nondecreasing, or the OWA criterion is median, then the problem is not at all approximable unless P=NP.

\section{Problem formulation}

Let $E=\{e_1,\dots,e_n\}$ be a finite set of \emph{elements} and $\Phi\subseteq 2^E$ be a set of \emph{feasible solutions}. In the deterministic case, each element $e_i\in E$ has a nonnegative \emph{cost} $c_i$ and we seek a solution whose total cost is minimum. Namely, we wish to solve the following optimization problem:
$$\mathcal{P}:\; \min_{X\in \Phi} F(X)=\min_{X\in \Phi} \sum_{e_i\in X} c_i$$ 
This formulation encompasses a large class of combinatorial optimization problems. In particular, for  the class of network problems  $E$ is the set of arcs of a given network $G=(V,E)$ and $\Phi$ contains the subsets of the arcs forming, for example,  $s-t$ paths, spanning trees, assignments, or $s-t$ cuts in $G$. In practice, problem~$\mathcal{P}$ is often expressed as a 0-1 programming one, where binary variable $x_i$ is associated with each element~$e_i$, $F(X)=\sum_{i=1}^n c_ix_i$,  and a system of constraints describes the set $\Phi$ in a compact form.

Before we discuss the uncertain version of problem $\mathcal{P}$, we recall the definition of the OWA operator, proposed by Yager in~\cite{YA88}. Let $(f_1,\dots,f_K)$ be a vector of reals. Let us introduce a vector $\pmb{w}=(w_1,\dots,w_K)$ such that $w_j\in [0,1]$, $j\in[K]$ (we use $[K]$ to denote the set $\{1,\dots,K\}$), and $w_1+\dots+w_K=1$. Let $\sigma$ be a permutation of $[K]$ such that $f_{\sigma(1)}\geq f_{\sigma(2)}\geq \dots\geq f_{\sigma(K)}$. Then
$${\rm owa}(f_1,\dots,f_K)=\sum_{i\in [K]} w_i f_{\sigma(i)}.$$
The OWA operator has several natural properties which easily follow from its definition  (see, e.g.~\cite{YKB11}). Since it is a convex combination of $f_1,\dots,f_K$ it holds $\min(f_1,\dots,f_K)\leq {\rm owa}(f_1,\dots,f_K)\leq \max(f_1,\dots,f_K)$.  It is also \emph{monotonic}, i.e. if $f_j \geq g_j$ for all $j\in [K]$, then 
$\mathrm{owa}(f_1,\dots,f_K)\geq \mathrm{owa}(g_1,\dots,g_K)$,
\emph{idempotent}, i.e. if $f_1=\dots=f_k=a$, then $\mathrm{owa}(f_1,\dots,f_K)=a$ and 
\emph{symmetric}, i.e. its value does not depend on the order of the values $f_1,\dots,f_K$. It generalizes several traditional criteria used in decision making under uncertainty and we will describe this fact later.

Assume that the costs in problem $\mathcal{P}$ are uncertain and they are specified in the form of scenario set $\Gamma=\{\pmb{c}_1,\dots,\pmb{c}_K\}$. Hence $\Gamma$ contains $K$ distinct \emph{cost scenarios}, where $\pmb{c}_j=(c_{1j},\dots,c_{nj})$ for $j\in [K]$. The cost of a given solution~$X$ depends on scenario~$\pmb{c}_j$ and 
will be denoted by
$F(X,\pmb{c}_j)=\sum_{e_i\in X} c_{ij}$. In this paper we will aggregate the costs by using the OWA operator. Namely, given a weight vector $\pmb{w}=(w_1,\dots,w_K)$, let us define
$$\mathrm{OWA}(X)={\rm owa}(F(X,\pmb{c}_1),\dots,F(X,\pmb{c}_K))=\sum_{j\in [K]} w_j F(X,\pmb{c}_{\sigma(j)}),$$
where $\sigma$ is a permutation of $[K]$ such that $F(X,\pmb{c}_{\sigma(1)})\geq F(X,\pmb{c}_{\sigma(2)} \geq \dots \geq F(X,\pmb{c}_{\sigma(K)})$. We will consider the following optimization problem:
$$\textsc{Min-Owa}~\mathcal{P}: \min_{X\in \Phi} \mathrm{OWA} (X).$$

We now discuss several special cases of the \textsc{Min-Owa}~$\mathcal{P}$ problem (see also Table~\ref{tabsc}). If $w_1=1$ and $w_j=0$ for $j=2,\dots,K$, then OWA becomes the maximum and the corresponding problem is denoted as  \textsc{Min-Max}~$\mathcal{P}$. This is a typical problem considered in the robust optimization framework. If $w_K=1$ and $w_j=0$ for $j=1,\dots,K-1$, then OWA becomes the minimum and the corresponding problem is denoted as \textsc{Min-Min}~$\mathcal{P}$. In general, if $w_k=1$ and $w_j=0$ for $j\in [K]\setminus\{k\}$, then OWA is the $k$-th largest cost and the problem is denoted as \textsc{Min-Quant}$(k)$~$\mathcal{P}$. In particular, when $k=\lfloor K/2 \rfloor +1$, then the $k$-th largest cost is median and the problem is denoted as \textsc{Min-Median}~$\mathcal{P}$.
If $w_j=1/K$ for all $j\in [K]$, i.e. when the  weights are \emph{uniform}, then OWA is the average (or the Laplace criterion) and the problem is denoted as \textsc{Min-Average}~$\mathcal{P}$. Finally, if $w_1=\alpha$ and $w_K=1-\alpha$, for some fixed $\alpha\in [0,1]$, and $w_j=0$ for the remaining weights, then we get the Hurwicz pessimism-optimism criterion and the problem is then denoted as \textsc{Min-Hurwicz}~$\mathcal{P}$.
\begin{table}[ht]
\caption{Special cases of \textsc{Min-Owa}~$\mathcal{P}$.} \label{tabsc}
\begin{tabular}{ll}
 \hline
	  Name of the problem & Weight distribution \\ \hline
		\textsc{Min-Max}~$\mathcal{P}$ & $w_1=1$ and $w_j=0$ for $j=2,\dots,K$ \\
		\textsc{Min-Min}~$\mathcal{P}$ & $w_K=1$ and $w_j=0$ for $j=1,\dots,K-1$ \\
		\textsc{Min-Average}~$\mathcal{P}$ & $w_j=1/K$ for $j\in [K]$ \\
		\textsc{Min-Quant}$(k)$~$\mathcal{P}$ & $w_k=1$ and $w_j=0$ for $j\in [K]\setminus \{k\}$ \\
		\textsc{Min-Median}~$\mathcal{P}$ & $w_{\lfloor K/2 \rfloor +1}=1$ and $w_j=0$ for $j\in [K] \setminus \{\lfloor K/2 \rfloor +1\}$\\
		\textsc{Min-Hurwicz}~$\mathcal{P}$ & $w_1=\alpha$, $w_K=1-\alpha$, $\alpha\in [0,1]$ and $w_j=0$ for $j\in [K]\setminus\{1,K\}$ \\ \hline
\end{tabular}
\end{table}

The aim of this paper is to
  explore the computational properties of $\textsc{Min-Owa}~\mathcal{P}$ depending on the number of scenarios and the weight distribution. In the next sections we will discuss the general problem as well as all its special cases listed in Table~\ref{tabsc}.

\section{Known complexity results}

Since \textsc{Min-Max}~$\mathcal{P}$ is a special case of \textsc{Min-Owa}~$\mathcal{P}$, all the known negative results for \textsc{Min-Max}~$\mathcal{P}$ remain true for \textsc{Min-Owa}~$\mathcal{P}$. We now briefly describe these results for various problems~$\mathcal{P}$. When $\mathcal{P}$ is \textsc{Shortest Path}, \textsc{Minimum Spanning Tree}, or \textsc{Minimum Assignment}, then \textsc{Min-Max}~$\mathcal{P}$ is  NP-hard for two scenarios~\cite{AV01, KY97}. Furthermore, when $\mathcal{P}$ is \textsc{Minimum s-t Cut}, then \textsc{Min-Max}~$\mathcal{P}$  is known to be strongly NP-hard for two scenarios~\cite{ABV08}. When the number of scenarios $K$ is \emph{unbounded}, i.e. 
$K$ is a part of the input, then the minmax versions of all these basic network problem become strongly NP-hard and not approximable within $O(\log^{1-\epsilon}K)$ for any $\epsilon>0$ unless NP $\subseteq$ 
DTIME$(n^{\mathrm{poly}( \log n)})$~\cite{KZ09, KZ11}. 
In the existing literature, the min-max version of the \textsc{Minimum Selecting Items} problem was also discussed. This problem has very simple combinatorial structure, and its set of  feasible solutions is defined as $\Phi=\{X\subseteq E: |X|=p\}$ for some fixed integer $p>0$. It turns out that \textsc{Min-Max Minimum Selecting Items} is  NP-hard for two scenarios~\cite{AV01} and becomes strongly NP-hard and hard to approximate within any constant factor if the number of scenarios is a part of the input~\cite{KZ13}.

The following positive and general result for \textsc{Min-Max}~$\mathcal{P}$ is well known (see, e.g.,~\cite{ABV09}):

\begin{thm}
\label{apprminmax}
	If $\mathcal{P}$ is polynomially solvable, then $\textsc{Min-Max}~\mathcal{P}$ is approximable within $K$.
\end{thm}
The idea of the $K$-approximation algorithm consists in solving the deterministic problem~$\mathcal{P}$ for the costs $\hat{c}_i=\max_{j\in [K]} c_{ij}$, $e_i\in E$. We thus first aggregate the costs using the maximum criterion and then compute an optimal solution for the aggregated costs. In this paper we will extend this idea to the general \textsc{Min-Owa}~$\mathcal{P}$ problem.  For particular problems~$\mathcal{P}$, better approximation algorithm exist. Namely, \textsc{Min-Max Minimum Spanning Tree} is approximable within $O(\log^2 K)$~\cite{KZ11} and \textsc{Min-Max Minimum Selecting Items} is approximable within $O(\log K/ \log \log K)$~\cite{DOE13}.

It is not difficult to identify  some special cases of \textsc{Min-Owa}~$\mathcal{P}$ which are polynomially solvable. 
\begin{obs}
If $\mathcal{P}$ is polynomially solvable, then $\textsc{Min-Min}~\mathcal{P}$ and $\textsc{Min-Average}~\mathcal{P}$ are polynomially solvable.
\end{obs}
Indeed,
in order to find an optimal solution to  $\textsc{Min-Average}~\mathcal{P}$ it is sufficient  
to solve~$\mathcal{P}$ for the average costs $\hat{c}_i=\frac{1}{K}\sum_{j\in [K]} c_{ij}$.
 In order to solve $\textsc{Min-Min}~\mathcal{P}$ it is enough to compute a sequence of solutions $X_1,\dots X_K$ such that $X_j$ minimizes $F(X,\pmb{c}_j)$ and choose $X_i\in\{X_1,\dots,X_K\}$ with the
minimum value of $F(X_i, \pmb{c}_i)$.

\section{The problem with two scenarios}
\label{sec2scen}

In this section we provide a characterization of the complexity of  \textsc{Min-Owa}~$\mathcal{P}$
 when the number of scenarios equals~2. This case can be described by a single weight $w_1\in [0,1]$, because $w_2=1-w_1$. Observe that OWA  is then equivalent to the Hurwicz criterion with $\alpha=w_1$. In this section, 
 for simplicity of notations,
 we will write~$\alpha$ instead of $w_1$.
 The case of polynomial solvability of $\textsc{Min-Owa}~\mathcal{P}$  is established by 
 the following theorem.
\begin{thm}
\label{thmK1}
		Let $K=2$. Then $\textsc{Min-Owa}~\mathcal{P}$ is polynomially solvable when $\mathcal{P}$ is polynomially solvable and $\alpha\in [0,1/2]$.
\end{thm}
\begin{proof}
		If $\alpha=0$, then we get the \textsc{Min-Min}~$\mathcal{P}$ problem which is polynomially solvable. So, assume that $\alpha>0$.
		Let us define
		$$H_1(X)=\max\{F(X,\pmb{c}_1), \alpha F(X,\pmb{c}_2)+(1-\alpha) F(X,\pmb{c}_1)\},$$
		$$H_2(X)=\max\{F(X,\pmb{c}_2), \alpha F(X,\pmb{c}_1)+(1-\alpha) F(X,\pmb{c}_2)\}.$$
	 An easy verification shows that $\mathrm{OWA}(X)=\min\{H_1(X),H_2(X)\}$. Let $X_1$ be a solution minimizing $\alpha F(X,\pmb{c}_2)+(1-\alpha) F(X,\pmb{c}_1)$ and let $X_2$ be a solution minimizing $\alpha F(X,\pmb{c}_1)+(1-\alpha) F(X,\pmb{c}_2)$. We will show that either $X_1$ or $X_2$ minimizes OWA. This will complete the proof, since both $X_1$ and $X_2$ can be computed in polynomial time provided that $\mathcal{P}$ is polynomially solvable.
	Let $X^*$ be an optimal solution to $\textsc{Min-Owa}~\mathcal{P}$ and suppose that 
	$\mathrm{OWA}(X^*)=H_1(X^*)\leq H_2(X^*)$. Then, by the definition of $X_1$, we get
    \begin{equation}
    \label{e00}
    \alpha F(X_1,\pmb{c}_2)+(1-\alpha)F(X_1,\pmb{c}_1)
    \leq \alpha F(X^*,\pmb{c}_2)+(1-\alpha)F(X^*,\pmb{c}_1)\leq H_1(X^*).
    \end{equation}
If  $F(X_1,\pmb{c}_1)\leq    \alpha F(X_1,\pmb{c}_2)+(1-\alpha)F(X_1,\pmb{c}_1)$,
then $H_1(X_1)\leq H_1(X^*)$ and $\mathrm{OWA}(X_1)\leq\mathrm{OWA}(X^*)$, which completes the proof.
Assume that $F(X_1,\pmb{c}_1)  >\alpha F(X_1,\pmb{c}_2)+(1-\alpha)F(X_1,\pmb{c}_1)$, which implies 
$F(X_1,\pmb{c}_1)  > F(X_1,\pmb{c}_2)$.
Since $\alpha \in (0,1/2]$, we get
\begin{equation}
\label{h00}
\alpha F(X_1,\pmb{c}_1)+(1-\alpha)F(X_1,\pmb{c}_2)\leq
\alpha F(X_1,\pmb{c}_2)+(1-\alpha)F(X_1,\pmb{c}_1).
\end{equation}
Furthermore
\begin{equation}
\label{h01}
F(X_1,\pmb{c}_2)\leq\alpha F(X_1,\pmb{c}_1)+(1-\alpha)F(X_1,\pmb{c}_2)=H_2(X_1).
\end{equation}
Inequalities (\ref{e00}), (\ref{h00}) and (\ref{h01}) imply
 $\mathrm{OWA}(X_1)\leq H_2(X_1)\leq H_1(X^*)=\mathrm{OWA}(X^*)$.
The second case, when $\mathrm{OWA}(X^*)=H_2(X^*)$ is just symmetric and involves $X_2$ instead of $X_1$.
\end{proof}
We now consider the case with $\alpha\in (1/2,1]$. We will show that it  is harder than the case with $\alpha\in [0,1/2]$, by using a slight modification of the proof of NP-hardness of the \textsc{Min-Max Shortest Path} problem for two scenarios shown in~\cite{JJ98, KY97}.
\begin{cor}
		\label{thm1a}
		Let $K=2$. Then for any $\alpha\in (1/2,1]$ the $\textsc{Min-Owa Shortest Path}$ problem is NP-hard.
\end{cor}
\begin{proof}
The reduction constructed in~\cite{JJ98, KY97} is as follows. Consider the following NP-complete \textsc{Partition} problem.  We are given a collection of positive integers $A=(a_1,\dots,a_n)$ such that $\sum_{i=1}^n a_i=2S$. We ask if  there is a subset $I\subseteq \{1,\dots,n\}$ such that $\sum_{i\in I} a_i=S$. Given an instance of \textsc{Partition}, we construct
 a graph shown in Figure~\ref{fig2}. 	
\begin{figure}[ht]
			\centering
      \includegraphics*{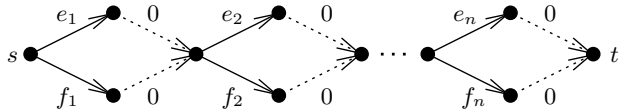}
      \caption{The graph in the reduction. The dummy (dashed) arcs have zero costs under 
       $\pmb{c}_1$ and $\pmb{c}_2$.} \label{fig2}
\end{figure}
We also form two scenarios. Under the first scenario $\pmb{c}_1$, the costs of the arcs $e_1,\dots,e_n$ are $a_1,\dots,a_n$ and the cost of all the remaining arcs are~0. Under the second scenario $\pmb{c}_2$, the costs of the arcs $f_1,\dots,f_n$ are $a_1,\dots,a_n$ and the costs of all the remaining arcs are~0.
Let $\alpha=1/2+\epsilon$, where $\epsilon\in (0,1/2]$.
 We claim that the answer to \textsc{Partition} is yes if and only if there is a path $X$ from $s$ to $t$ such that $\mathrm{OWA}(X)\leq S$. Indeed, if the answer is yes, then we form the path $X$ by choosing arcs $e_i$ for $i\in I$ and $f_i$ for $i\notin I$ and complete it by dummy arcs. Then $F(X,\pmb{c}_1)=F(X,\pmb{c}_2)=S$ and $\mathrm{OWA}(X)=S$. On the other hand, suppose that the answer to \textsc{Partition} is no. Then for each path $X$ either $F(X,\pmb{c}_1)=S_1>S$ or $F(X,\pmb{c}_2)=S_2>S$. Assume that the first case holds (the second one is symmetric). Then $F(X,\pmb{c}_2)=2S-S_1$ and $\mathrm{OWA}(X)=(\frac{1}{2}+\epsilon)S_1+(\frac{1}{2}-\epsilon)(2S-S_1)=S+2\epsilon(S_1-S)$ and so $\mathrm{OWA}(X)>S$ since $S_1>S$ and $\epsilon>0$.
\end{proof}

Theorem~\ref{thm1a} remains  true when $\mathcal{P}$ is  \textsc{Minimum Spanning Tree}, \textsc{Minimum s-t Cut} or \textsc{Minimum Assignment}. To see this, observe that each path in the graph shown in Figure~\ref{fig2} can be transformed into a spanning tree of the same cost under both scenarios by adding a number of dummy arcs and vice versa, each spanning tree in this graph can be transformed into a path of the same cost under both scenarios by removing a number of dummy arcs. In order to prove the result for \textsc{Minimum s-t cut} and \textsc{Minimum Assignment}, we only need to replace the graph from Figure~\ref{fig2} with the graphs depicted
 in Figure~\ref{fig3}a and~\ref{fig3}b, respectively. The proof is then the same as for the \textsc{Shortest Path} problem.
 Therefore, from now on each negative result proven for the \textsc{Shortest Path} problem, can be transformed into \textsc{Minimum Spanning Tree}, \textsc{Minimum s-t Cut} or \textsc{Minimum Assignment} by using the transformation just described.

\begin{figure}[ht]
\centering
      \includegraphics*{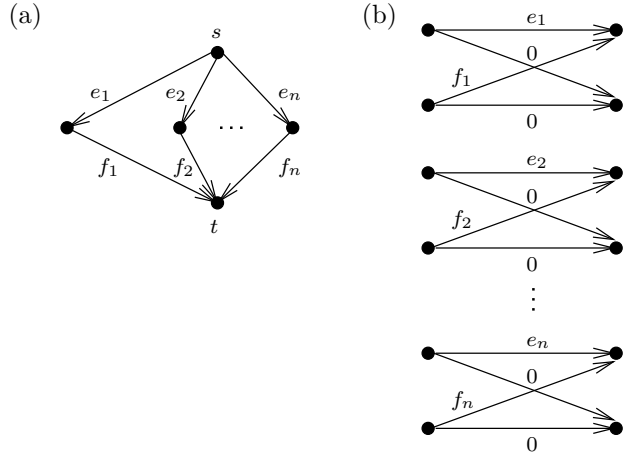}
      \caption{The graphs: (a) for the minimum $s-t$ cut problem,
       (b) for the minimum assignment problem.} \label{fig3}
\end{figure}

In Section~\ref{secunb} we will show that the problem with $K=2$ and $\alpha\in (1/2,1]$ admits a simple $2\alpha$-approximation algorithm, provided that $\mathcal{P}$ is polynomially solvable.
Moreover, we will prove  that when $K=3$ minimizing the Hurwicz criterion for \textsc{Shortest path} is NP-hard for any $\alpha\in (0,1]$.

\section{The problem with constant number of scenarios}
\label{secfptas}

In this section we discuss the case when $K$ is constant. We will show that under some additional  assumptions  $\textsc{Min-Owa}~\mathcal{P}$ admits then a fully polynomial time approximation scheme (FPTAS), i.e. a family of $(1+\epsilon)$-approximation algorithms which are polynomial in the input size and $1/\epsilon$, $\epsilon>0$. In order to construct the FPTAS, we will use the results obtained in~\cite{PY00} and~\cite{MS08}.

Let us fix $\epsilon>0$ and let $P_{\epsilon}(\Phi)$ be the set of solutions such that for all $X\in \Phi$, there is $Y\in P_{\epsilon}(\Phi)$ such that $F(Y,\pmb{c}_j)\leq (1+\epsilon)\,F(X,\pmb{c}_j)$ for all $j\in [K]$.
 We now recall the definition  of an \emph{exact problem} associated with~$\mathcal{P}$ (see~\cite{MS08}). 
 Given a vector $(v_1,\dots, v_K)$, we ask if there  is a solution $X\in \Phi$ such
  that $F(X,\pmb{c}_j)=v_j$ for all $j\in [K]$.
   Basing on the results obtained in~\cite{PY00}, it was proven in~\cite{MS08} that
if the exact problem associated with~$\mathcal{P}$ can be solved in pseudopolynomial time, then for any $\epsilon>0$, the set $P_{\epsilon}(\Phi)$ can be determined in time polynomial in the input size and $1/\epsilon$. This implies the following result:
\begin{thm}
\label{thmfptas}
If the exact problem associated with~$\mathcal{P}$ can be solved in pseudopolynomial time, then $\textsc{Min-Owa}~\mathcal{P}$ admits an FPTAS.
\end{thm}
\begin{proof}
	Let us fix $\epsilon>0$ and let $Y$ be a solution of the minimum value of $\mathrm{OWA}(Y)$ among all the solutions in $P_{\epsilon}(\Phi)$. From the  results obtained in~\cite{MS08, PY00}, it follows that we can find $Y$ in time polynomial in the input size and $1/\epsilon$. Assume that $X^*$ is an optimal solution to \textsc{Min-OWA}~$\mathcal{P}$. Define vector $\pmb{v}^*=((1+\epsilon)F(X^*,\pmb{c}_1),\dots,(1+\epsilon)F(X^*,\pmb{c}_K))$. By the definition of $Y$ we get $F(Y,\pmb{c}_j)\leq (1+\epsilon)F(X^*,\pmb{c}_j)$ for all $j\in [K]$. The monotonicity of  OWA  implies ${\rm OWA(Y)}\leq {\rm owa}(\pmb{v}^*)=(1+\epsilon){\rm OWA}(X^*)$.  We have thus obtained an FPTAS for $\textsc{Owa}~\mathcal{P}$. 
\end{proof}

It turns out that the exact problem associated with~$\mathcal{P}$ can be solved in pseudopolynomial time for some particular problems~$\mathcal{P}$, provided that 
the number of scenarios~$K$ is constant.
This is the case for \textsc{Shortest Path},  \textsc{Minimum Spanning Tree} and some other problems described, for example, in~\cite{ABV10}. However,
it is worth pointing out that
  the running time of the FPTAS's obtained is exponential in $K$, so their practical applicability is limited to very small values of $K$. In the next section we construct approximation algorithms which are much faster and can be applied to problems with large number of scenarios.

\section{The problem with unbounded scenario set}
\label{secunb}

In this section we examine  the case, when the number of scenarios is unbounded, i.e. it is a part of the input. We discuss the  complexity and approximability of 
$\textsc{Min-Owa}~\mathcal{P}$ depending on various weight distributions. As we know from the results for \textsc{Min-Max}~$\mathcal{P}$, \textsc{Min-Owa}~$\mathcal{P}$  is not approximable within any constant factor for many basic problems $\mathcal{P}$, for example when $\mathcal{P}$ is \textsc{Shortest Path}. However, we can try to construct  approximation algorithms whose worst case ratio is a function of the number of scenarios~$K$. It turns out that 
the existence of such algorithms depends on the ordering of weights in the OWA operator, i.e. whether the weights are nonincreasing or nondecreasing. We thus study first these two types of weight distributions.

\subsection{Nonincreasing weights}
\label{secnonin}

Suppose  that the weights are nonincreasing, i.e. $w_1\geq w_2\geq \dots \geq w_K$. Notice that this case contains both the maximum and the average criteria  as special and boundary cases. Furthermore, it holds $w_1\geq 1/K$, because the weights must sum up to~1. The nonincreasing weights can be used if the idea of robust optimization is adopted. Namely, a decision maker assigns larger weights to larger solution costs. In the extreme case this leads to the maximum criterion, where only the largest solution cost is taken into account. The analysis of the case with~2 scenarios (Section~\ref{sec2scen}) shows that  \textsc{Min-Owa Shortest Path}  is NP-hard for all nonincreasing weight distributions except for the uniform one, when the weights are equal.

We  now construct an approximation algorithm for \textsc{Min-Owa}~$\mathcal{P}$ whose idea is to aggregate the costs of each element $e_i\in E$ by using the OWA operator and compute then an optimal solution for the aggregated costs. Consider element $e_i\in E$ and let 
$\hat{c}_{i1}\geq \hat{c}_{i2}\geq \dots \geq \hat{c}_{iK}$ be the ordered costs of $e_i$.  Let $\hat{c}_i=\sum_{j\in [K]} w_j \hat{c}_{ij}$ be the aggregated cost of $e_i$ and $\hat{C}(X)=\sum_{e_i\in X} \hat{c}_i$. Let $\hat{X}$ be a solution minimizing $\hat{C}(X)$. Of course, $\hat{X}$ can be computed in polynomial time if $\mathcal{P}$ is polynomially solvable.  
The following theorem is a generalization of Theorem~\ref{apprminmax}.
\begin{thm}
\label{thm1}
If the weights are nonincreasing, then $\mathrm{OWA}(\hat{X})\leq w_1 K\cdot\mathrm{OWA}(X)$ for any $X\in \Phi$ and the bound it tight.
\end{thm}
\begin{proof}
Let $\sigma$ be a sequence of $[K]$ such that $F(\hat{X},\pmb{c}_{\sigma(1)})\geq \dots \geq F(\hat{X},\pmb{c}_{\sigma(K)})$.
 From the definition of the OWA operator and the assumption that the weights are nonincreasing, 
 we obtain:
\begin{equation}
\label{f1} 
\mathrm{OWA}(\hat{X})=\sum_{j\in [K]} w_j \sum_{e_i\in \hat{X}} c_{i\sigma(j)}=\sum_{e_i\in \hat{X}} \sum_{j\in [K]} w_j c_{i\sigma(j)}\leq \sum_{e_i\in \hat{X}} \sum_{j\in [K]} w_j \hat{c}_{ij}=\hat{C}(\hat{X}).
\end{equation} 	
From the definition of $\hat{X}$ and the fact that $w_1$ is the largest weight we obtain:
\begin{equation}
\label{f2}
\hat{C}(\hat{X})\leq \hat{C}(X)=\sum_{e_i\in X}\sum_{j\in[K]}w_j\hat{c}_{ij}\leq w_1\sum_{e_i\in X} \sum_{j\in [K]}  c_{ij}
\end{equation} and, again from the assumption that the weights are nonincreasing we get:
\begin{equation}
\label{f3}
\mathrm{OWA}(X)\geq \sum_{j\in [K]} \frac{1}{K} F(X,\pmb{c}_{\sigma(j)})=\frac{1}{K} \sum_{e_i\in X}\sum_{j\in [K]} c_{ij}.
\end{equation}
Finally, combining~(\ref{f1}), (\ref{f2}) and~(\ref{f3}) yields $\mathrm{OWA}(\hat{X})\leq w_1 K\cdot \mathrm{OWA}(X)$. 

In order to prove that the bound is tight consider the problem where $E=\{e_1,\dots, e_{2K}\}$ and $\Phi=\{X\subseteq E: |X|=K\}$. The cost scenarios are shown in Table~\ref{tab1}. 

\begin{table}[ht]
\centering
\caption{A hard example for the approximation algorithm.} \label{tab1}
\begin{tabular}{l|lllllll}
   & $\pmb{c}_1$ & $\pmb{c}_2$ & $\pmb{c}_3$ & $\dots$ & $\pmb{c}_K$ \\ \hline
$e_1$ & 0 & 0 & 0 & $\dots$ & $1$ \\
$e_2$ & 0 & 0 & 0 & $\dots$ & $1$ \\
$\vdots$ \\
$e_K$ & 0 & 0 & 0 & $\dots$ & $1$ \\ \hline
$e_{K+1}$ & $1$ & 0 & 0 & $\dots$ & 0 \\
$e_{K+2}$ & 0 & $1$ & 0 & $\dots$ & 0\\
$\vdots$ \\
$e_{2K}$ & 0 & 0 & 0 & $\dots$ & $1$
\end{tabular}
\end{table}
Observe that all the elements have the same aggregated costs for any weights $w_1,\dots, w_K$. Hence, we may choose any feasible solution as $\hat{X}$. If $\hat{X}=\{e_1,\dots,e_K\}$, then $\mathrm{\rm OWA}(\hat{X})=w_1K$. But if $X=\{e_{2K+1},\dots,e_{2K}\}$, then $\mathrm{OWA}(X)=\sum_{j\in [K]} w_j =1$ and so 
$\mathrm{OWA}(\hat{X})=w_1 K\cdot {\rm OWA}(X)$.
\end{proof}
Theorem~\ref{thm1} leads to the following corollary:
\begin{cor}
\label{coro1}
	If the weights are nonincreasing and $\mathcal{P}$ is polynomially solvable, then \textsc{Min-Owa}~$\mathcal{P}$ is approximable within $w_1K$.
\end{cor}
Let us focus on some consequences of  Corollary~\ref{coro1}. 
Since the weights are nonincreasing, $w_1\in[1/K,1]$. Thus,
if $w_1=1$, i.e. when OWA becomes the maximum, we get the $K$-approximation algorithm, which is known in the literature
(see, e.g.,~\cite{ABV09}). On the other hand, if $w_1=1/K$, i.e. 
when  OWA becomes the average, $\hat{X}$ is an optimal solution to \textsc{Min-Owa}~$\mathcal{P}$. Therefore,  the more uniform is the weight distribution the better is the approximation ratio of the algorithm. 

In the proof of Theorem~\ref{thm1} we have assumed that we are able to solve the deterministic problem $\mathcal{P}$ in polynomial time. Of course, this is not true  for many combinatorial optimization problems which are NP-hard even in the deterministic case. However, in this case we often know a $\gamma$-approximation algorithm for $\mathcal{P}$, for some $\gamma>1$. We can  modify inequalities~(\ref{f2}) and write $\hat{C}(\hat{X})\leq \gamma w_1 \sum_{e_i\in X}\sum_{j\in [K]} c_{ij}$. As a result we get ${\rm OWA}(\hat{X})\leq w_1\gamma K \cdot {\rm OWA}(X)$ for any $X\in \Phi$, which leads to the following corollary:
\begin{cor}
If the weights are nonincreasing and $\mathcal{P}$ is approximable within 
$\gamma>1$, then \textsc{Min-Owa}~$\mathcal{P}$ is approximable within $w_1\gamma K$.
\end{cor}

\subsection{Nondecreasing weights}
\label{secnond}

Assume now that the weights are nondecreasing, i.e. $w_1\leq w_2\leq \dots \leq w_K$. Notice that this case contains both the minimum and average criteria as special cases. The following theorem shows that this case is much harder than the one with nonincreasing weights.
\begin{thm}
\label{thm2}
	Assume that the weights are nondecreasing and $K$ is unbounded. Then \textsc{Min-Owa Shortest Path} is not at all approximable  unless $P=NP$.
\end{thm}
\begin{proof}
	We make use of the following \textsc{Min 3-Sat} problem, which is known to be NP-complete~\cite{AZ02, KM94}. We are given boolean variables $x_1,\dots,x_n$ and a collection of clauses $C_1,\dots, C_m$, where each clause is a disjunction of at most three literals (variables or their negations). We ask if there is a 0-1 assignment to the variables which satisfies at most $L$ clauses. Given an instance of \textsc{Min 3-Sat}, we construct the graph shown in Figure~\ref{fig2} -- the same graph as in the proof of Theorem~\ref{thm1a}. The arcs $e_1,\dots,e_n$ correspond to literals $x_1,\dots,x_n$ and the arcs $f_1,\dots,f_n$ correspond to literals $\overline{x}_1,\dots,\overline{x}_n$. There is one-to-one correspondence between paths from $s$ to $t$ and 0-1 assignments to the variables. We fix $x_i=1$ if a path chooses $e_i$ and $x_i=0$ if a path chooses $f_i$.
The set $\Gamma$ is constructed as follows. For each clause $C_j=(l^j_1\vee l^j_2 \vee l^j_3)$,
$j\in [m]$, we form the cost scenario $\pmb{c}_j$ in which the costs of the arcs corresponding to
$l^j_1$, $l^j_2$ and $l^j_3$ are set to~$1$ and the costs of the remaining arcs are set to~0. We fix $w_1=\dots=w_{L}=0$ and $w_{L+1}=\dots=w_K=1/(K-L)$, where $K=m$. Notice that the weights are nondecreasing. Suppose that the answer to \textsc{Min 3-Sat} is yes. Then there is an assignment satisfying at most $L$ clauses. Consider the path $X$ corresponding to this assignment. From the construction of $\Gamma$ it follows that the cost of $X$ is positive under at most $L$ scenarios. In consequence $\mathrm{OWA}(X)=0$. On the other hand, if the answer to \textsc{Min 3-Sat} is no, then any assignment satisfies more than $L$ clauses and each path $X$ has a positive cost (not less than one) for more than $L$ scenarios. This implies $\mathrm{OWA}(X)\geq 1$ for all $X\in \Phi$. Accordingly to the above, we have: the answer to \textsc{Min 3-sat} is yes if and only if there is a path $X$ such that $\mathrm{OWA}(X)=0$. Hence the problem is not at all approximable unless P=NP.

This negative result remains true even if the element costs under all scenarios are positive. To see this it is enough to modify the construction of the scenario set as follows. Under scenario $\pmb{c}_j$, the costs of the arcs corresponding to the literals $l^j_1$, $l^j_2$ and $l^j_3$ are set to~$(n+1)(K-L)\rho(|I|)$, for any polynomially computable function $\rho(|I|)$ of the input size $|I|$, and the costs of the remaining arcs are set to~1. Now, if the answer to~\textsc{Min~3-Sat} is yes then there is a path $X$ such that ${\rm OWA}(X)\leq n$, and if the answer is no, then for all paths $X$ it holds ${\rm OWA}(X)\geq (n+1)\rho(|I|)$. Consequently, the gap is $\rho(|I|)$ and no $\rho(|I|)$-approximation for the problem exists unless P=NP.

\end{proof}

It follows from Theorem~\ref{thm2} that for the general \textsc{Min-Owa}~$\mathcal{P}$ problem there is no approximation algorithm with a worst case ratio bounded by a polynomially computable function of $K$. This is contrary to the special case of the problem with nonincreasing weights, where such an algorithm exists (see Section~\ref{secnonin}).

\subsection{The $k$th largest cost criterion}

In some applications, we wish to minimize the $k$th largest solution cost, in particular the median when $k=\lfloor K/2 \rfloor+1$. This leads to the \textsc{Min-Quant}$(k)$~$\mathcal{P}$ and \textsc{Min-Median}~$\mathcal{P}$ problems, respectively. The complexity of \textsc{Min-Quant}$(k)$~$\mathcal{P}$ depends on two parameters, namely $k$ and $K$, which can be constant or unbounded. It is clear the \textsc{Min-Quant(1)}~$\mathcal{P}$ is NP-hard and \textsc{Min-Quant($K$)}~$\mathcal{P}$ is polynomially solvable, for example when $\mathcal{P}$ is the \textsc{Shortest Path}, since the former is  \textsc{Min-Max}~$\mathcal{P}$ and the latter is \textsc{Min-Min}~$\mathcal{P}$. It is easy to show that \textsc{Min-Max}~$\mathcal{P}$ with $K$ scenarios is equivalent to \textsc{Min-Quant($k$)}~$\mathcal{P}$ with $K+k-1$ scenarios, where the first $K$ scenarios are the same as in \textsc{Min-Max}~$\mathcal{P}$ and under the remaining $k-1$ scenarios all elements have sufficiently large costs. Thus, in particular, \textsc{Min-Quant($k$) Shortest Path} is NP-hard for any constant $K\geq 2$ and any constant $k\in\{1,\dots,K-1\}$.  From the results obtained in Section~\ref{secfptas}, we know that  \textsc{Min-Quant($k$)}~$\mathcal{P}$ admits an FPTAS when $K$ is constant and the corresponding exact problem can be solved in pseudopolynomial time. 

We now investigate the case when $K$ is unbounded and $k$ is constant.  Observe that \textsc{Min-Quant($k$)}~$\mathcal{P}$ can be reduced to solving a family of $\binom{K}{k-1}$ \textsc{Min-Max}~$\mathcal{P}$ problems. It follows from the fact, that we can enumerate all subsets of $k-1$ scenarios, and for each such a subset, say $\Gamma'$,  we can compute an optimal solution to the corresponding \textsc{Min-Max}~$\mathcal{P}$ problem with scenario set $\Gamma \setminus \Gamma'$. One of the solutions computed must be optimal for \textsc{Min-Quant($k$)}~$\mathcal{P}$. In consequence, if \textsc{Min-Max}~$\mathcal{P}$ is approximable within $\gamma$, then the same result holds for \textsc{Min-Quant($k$)}~$\mathcal{P}$, provided that $k$ is constant. Since \textsc{Min-Max}~$\mathcal{P}$ is approximable within $K$, when $\mathcal{P}$ is polynomially solvable, we get the following result:
 
\begin{cor}
	If $\mathcal{P}$ is polynomially solvable and $k$ is constant, then \textsc{Min-Quant($k$)}~$\mathcal{P}$ is approximable within $K$.
\end{cor}

The approximation algorithm is efficient only when $k$ is close to 1 or close to $K$. Its running time is not polynomial when $k$ is unbounded, because $\binom{K}{k-1}$ is then exponential in $k$. 

We now study the case when both $K$ and $k$ are unbounded. 
Observe that this is the case, for example, when $K$ is unbounded and OWA is  median, because $k=\lfloor K/2 \rfloor +1$ is then a function of $K$. We prove the following negative result:
\begin{thm}
\label{thmmed}
	Let $K$ be unbounded. Then \textsc{Min-Median}
	 \textsc{Shortest Path} is not at all approximable unless $P=NP$.
\end{thm}
\begin{proof}
	The reduction is very similar to that in the proof of Theorem~\ref{thm2}. It is enough to modify it as follows.
	Assume first that  $L<\lfloor m/2 \rfloor$. We then add to $\Gamma$ additional $m-2L$ scenarios with the costs equal to~1 for all the arcs. So the number of scenarios is $2m-2L$. We fix $w_{m-L+1}=1$ and $w_j=0$ for the remaining scenarios. Now, the answer to \textsc{Min 3-SAT} is yes, if and only if there is a path $X$ whose cost is 1 under at most $L+m-2L=m-L$ scenarios  and equivalently  $\mathrm{OWA}(X)=0$, which is due to the definition of the weights.
	Assume that $L>\lfloor m/2 \rfloor$. We then we add to $\Gamma$ additional $2L-m$ scenarios with the costs equal to~0 for all the arcs.  The number of scenarios is then $2L$. We fix $w_{L+1}=1$ and $w_j=0$ for all the remaining scenarios. Now, the answer to \textsc{Min 3-SAT} is yes, if and only if there is a path $X$ whose cost is 1 under at most $L$ scenarios. According to the definition of the weights, it is equivalent to $\mathrm{OWA}(X)=0$. We thus can see that it is NP-hard to check whether there is a path $X$ such that ${\rm OWA}(X)\leq 0$ and the theorem follows. Using a reasoning similar to that in the proof of Theorem~\ref{thm2}, we can show that the negative result remains true when all elements have positive costs under all scenarios.
\end{proof}

 Theorem~\ref{thmmed} states that there is no approximation 
algorithm for \textsc{Min-Quant}$(k)$~$\mathcal{P}$ whose worst case ratio is a polynomially computable function of $K$ and $k$. In consequence, minimizing the $k$th largest cost can be much harder than minimizing the largest cost.

\subsection{The Hurwicz criterion}

In Section~\ref{sec2scen} we have proved that when the number of scenarios equals~2, $\mathcal{P}$ is polynomially solvable and $\alpha\in [0,1/2]$, then \textsc{Min-Hurwicz}~$\mathcal{P}$ is polynomially solvable. We now show that this is no longer true when the number of scenarios is greater than~2. 
\begin{obs}
\label{thmhurr2} For any $\alpha\in(0,1]$,
there is a polynomial time approximation preserving reduction from
 \textsc{Min-Max}~$\mathcal{P}$ with $K$ scenarios to 
 \textsc{Min-Hurwicz}~$\mathcal{P}$ with $K+1$ scenarios.
\end{obs}
\begin{proof}
Consider an instance of \textsc{Min-Max}~$\mathcal{P}$ with 
scenario set $\Gamma=\{\pmb{c}_1,\ldots,\pmb{c}_K\}$.
To build an instance of  \textsc{Min-Hurwicz}~$\mathcal{P}$, we only 
add to $\Gamma$ the $(K+1)$th
scenario $\pmb{c}_{K+1}$ with the costs equal to~0 for all the elements.
Now for each $X\in \Phi$ it holds
$\mathrm{OWA}(X)=\alpha\cdot \max_{j\in [K+1]} F(X,\pmb{c}_j)=
\alpha\cdot \max_{j\in [K]} F(X,\pmb{c}_j)$.
Therefore, it is evident that the reduction is approximation preserving.
\end{proof}
Theorem~\ref{thmhurr2} and the hardness results obtained by \cite{KZ09, KY97},
lead to the following corollary:
\begin{cor}
For any $\alpha\in (0,1]$ and $K\geq 3$ the \textsc{Min-Hurwicz Shortest Path} problem is NP-hard.
Furthermore, if $K$ is unbounded, then for any $\alpha\in (0,1]$,  \textsc{Min-Hurwicz Shortest Path} is
strongly NP-hard and not approximable within $O(\log^{1-\epsilon}K)$ for any $\epsilon>0$ unless NP $\subseteq$
 DTIME$(n^{\mathrm{poly}(\log n)})$.
\label{corhurk1}
\end{cor}

We now construct two approximation algorithms for \textsc{Min-Hurwicz}~$\mathcal{P}$ which can be applied when $K$ is unbounded. 
Notice that the approximation algorithm designed in Section~\ref{secnonin} cannot be applied to the case with $K\geq 3$ since the weights are then not nonincreasing. The first algorithm can be applied when $\alpha\in[1/2,1]$ and the second one will be valid for any $\alpha \in (0,1]$.
Let $\hat{c}_{i1}\geq \dots \geq \hat{c}_{iK}$ be the ordered sequence of the costs of element $e_i\in E$ over all scenarios. Let $\hat{c}_i=\alpha \hat{c}_{i1}+(1-\alpha)\hat{c}_{i2}$, $\hat{C}(X)=\sum_{e_i\in X} \hat{c}_i$ and let $\hat{X}$ minimize $\hat{C}(X)$. 
\begin{thm}
\label{thmhurr1}
	If $\alpha\in [1/2,1]$ and $K\geq 2$, then  for any $X\in \Phi$ the following inequality holds:
	$$\mathrm{OWA}(\hat{X})\leq   [\alpha K+(1-\alpha)(K-2)]\mathrm{OWA}(X).$$
\end{thm}

\begin{proof}
Let $\sigma$ be a permutation  of $[K]$ such that $F(\hat{X},\pmb{c}_{\sigma(1)})\geq \dots \geq F(\hat{X},\pmb{c}_{\sigma(K)})$. It holds
\begin{equation}
\label{g1} 
\mathrm{OWA}(\hat{X})=\sum_{e_i\in \hat{X}} (\alpha c_{i\sigma(1)}+(1-\alpha)c_{i\sigma(K)})\leq \sum_{e_i\in \hat{X}} \hat{c}_{i}=\hat{C}(\hat{X}).
\end{equation} 	
Since $\alpha \geq 1/2$ and $\hat{X}$ minimize $\hat{C}(X)$, we get
\begin{equation}
\label{g2}
\hat{C}(\hat{X})\leq \hat{C}(X)=\sum_{e_i\in X}\hat{c}_{i}\leq \alpha \sum_{e_i\in X} \sum_{j\in [K]}  c_{ij}.
\end{equation} 
We now prove the following inequality:
\begin{equation}
\label{g3}
\mathrm{OWA}(X)\geq \frac{1}{K+\frac{1-\alpha}{\alpha}(K-2)} \sum_{e_i\in X}\sum_{j\in [K]} c_{ij}.
\end{equation}
Let $\rho$ be a permutation of $[K]$ such that $F(X,\pmb{c}_{\rho(1)})\geq \dots \geq F(X,\pmb{c}_{\rho(K)})$.
Observe first that
\begin{equation}
\label{g400}
\sum_{e_i\in X}\sum_{j\in [K]} c_{ij} =\sum_{j \in [K]} F(X,\pmb{c}_{\rho(j)})\leq (K-1)F(X,\pmb{c}_{\rho(1)})+F(X,\pmb{c}_{\rho(K)}).
\end{equation}
We  will now show that for any $\alpha\in [1/2,1]$ it holds
\begin{equation}
{\rm OWA}(X)\geq \frac{K-1}{K+\frac{1-\alpha}{\alpha}(K-2)} F(X,\pmb{c}_{\rho(1)})+\frac{1}{K+\frac{1-\alpha}{\alpha}(K-2)} F(X,\pmb{c}_{\rho(K)}),
\label{g4}
\end{equation}
which together with~(\ref{g400}) will imply~(\ref{g3}).
Since $K\geq 2$, and ${\rm OWA}(X)=\alpha F(X,\pmb{c}_{\rho(1)})+(1-\alpha)F(X,\pmb{c}_{\rho(K)})$,
 the inequality (\ref{g4}) can be rewritten in the following equivalent form:
\begin{equation}
\alpha(2\alpha-1)F(X,\pmb{c}_{\rho(1)})+
((1-\alpha)K-2(1-\alpha)^2-\alpha)F(X,\pmb{c}_{\rho(K)})\geq 0.
\label{g5}
\end{equation}
Note that $F(X,\pmb{c}_{\rho(1)})\geq F(X,\pmb{c}_{\rho(K)})\geq 0$.
Hence, in order to prove (\ref{g5}), it suffices to show it for $K=2$. Thus, we get
\begin{equation}
\alpha(2\alpha-1)(F(X,\pmb{c}_{\rho(1)})
-F(X,\pmb{c}_{\rho(K)}))\geq 0.
\label{g6}
\end{equation}
We  see  at once that inequality~(\ref{g6}) holds for every $\alpha\in [1/2,1]$.
Combining~(\ref{g1}), (\ref{g2}) and~(\ref{g3}) completes the proof.
\end{proof}

\begin{cor}
\label{corhurk}
		If $\alpha\in [1/2,1]$ and $\mathcal{P}$ is polynomially solvable, then \textsc{Min-Hurwicz}~$\mathcal{P}$ is approximable within $\alpha K + (1-\alpha)(K-2)$.
\end{cor}

Let us analyze some consequences of Corollary~\ref{corhurk}. If $K=2$, then the algorithm is equivalent to the approximation algorithm designed in Section~\ref{secnonin}, in which case we get the approximation ratio of $2\alpha$. The largest worst case ratio of the algorithm, equal to $K$, occurs when $\alpha=1$, i.e. when the Hurwicz criterion becomes the maximum.  On the other hand, the smallest worst case ratio equal to $K-1$ is when $\alpha=1/2$, i.e. when the Hurwicz criterion is the average of the minimum and maximum.

The bound obtained in Theorem~\ref{thmhurr1} does not hold when $\alpha\in (0, 1/2)$. For this case we will design an approximation algorithm which is based on a different idea. Suppose that we have a $\gamma$-approximation algorithm for the \textsc{Min-Max}~$\mathcal{P}$ problem. Notice that in the general case, when $\mathcal{P}$ is polynomially solvable, $\gamma$ can be equal to $K$ (see Corollary~\ref{coro1}), but for some particular problems such as \textsc{Min-Max Minimum Spanning Tree}  or
 \textsc{Min-Max Minimum  Selecting Items}
better approximation algorithms exist (\cite{DOE13, KZ11, KZ09x, KZ13}). 
\begin{thm}
\label{thmhurr3}
Suppose that there exists an approximation algorithm
for \textsc{Min-Max}~$\mathcal{P}$
 with a worst case ratio of~$\gamma>1$.
Let $\hat{X}\in \Phi$ be a solution constructed by this algorithm. Then for any $\alpha \in (0,1]$ and $X\in \Phi$ it holds
\begin{equation}
\mathrm{OWA}(\hat{X})\leq 
\begin{cases}
	\gamma\cdot {\rm OWA}(X) &\text{\rm if $\alpha=1$ or $\min_{j\in [K]} F(\hat{X},\pmb{c}_j)=0$,}\\
	 \frac{\gamma}{\alpha} \cdot {\rm OWA}(X) &\text{\rm if $\alpha\in (0,1).$} 
\end{cases}
\label{apprhur}
\end{equation}
\end{thm}
\begin{proof}
It follows immediately that for any $\alpha\in (0,1]$:
\begin{equation}
\mathrm{OWA}(X)\geq \alpha \max_{j\in [K]} F(X,\pmb{c}_j)
\geq \alpha\cdot OPT_{\max},
\label{lbhur}
\end{equation}
where $OPT_{\max}=\min_{X\in \Phi}\max_{j\in [K]} F(X,\pmb{c}_j)$.
For the first case of~(\ref{apprhur}), we note that
\[
\mathrm{OWA}(\hat{X})=\alpha \cdot \max_{j\in [K]} F(\hat{X},\pmb{c}_j)
\leq \gamma\alpha \cdot OPT_{\max}{\leq}\gamma\cdot {\rm OWA}(X),
\]
where the last inequality follows from~(\ref{lbhur}). For the case $\alpha\in (0,1)$, we get
\[
\mathrm{OWA}(\hat{X})=\alpha \cdot \max_{j\in [K]} F(\hat{X},\pmb{c}_j)+
(1-\alpha) \cdot \min_{j\in [K]} F(\hat{X},\pmb{c}_j)\leq
\max_{j\in [K]} F(\hat{X},\pmb{c}_j)\leq \gamma \cdot OPT_{\max}
{\leq}\frac{\gamma}{\alpha}\cdot {\rm OWA}(X),
\]
where the last inequality also follows from~(\ref{lbhur}).
\end{proof}

Let us now apply Corollary~\ref{corhurk} and Theorem~\ref{thmhurr3} to some special cases of \textsc{Min-Hurwicz}~$\mathcal{P}$. If $\mathcal{P}$ is \textsc{Shortest Path}, then the problem is approximable within $\alpha K +(1-\alpha) (K-2)$ for $\alpha \in [1/2,1]$ and within $K/\alpha$ for $\alpha \in (0,1/2)$, if we use the $K$-approximation algorithm for the \textsc{Min-Max Shortest Path} problem. If $\mathcal{P}$ is \textsc{Minimum Spanning Tree}, then the problem is approximable within $O((1/\alpha)\log^2 K)$ with a high probability for any $\alpha\in (0,1]$, if we use the randomized $O(\log^2 K)$-approximation algorithm for \textsc{Min-Max Minimum Spanning Tree} designed in~\cite{KZ11}. Finally, when $\mathcal{P}$ is \textsc{Minimum Selecting Items}, then the problem is approximable within $O((1/\alpha) \log K/\log\log K)$, when the $O(\log K / \log \log K)$-approximation algorithm for \textsc{Min-Max Minimum Selecting Items} constructed in~\cite{DOE13} is applied.

\section{Summary}

In this paper we have discussed a class of combinatorial optimization problems with uncertain costs specified in the form of a discrete scenario set.  We have applied the OWA operator as the criterion of choosing a solution.
We have obtained several 
general computational properties of the resulting
 \textsc{Min-Owa}~$\mathcal{P}$ problem.
 Except for some very special weight distributions, the \textsc{Min-Owa}~$\mathcal{P}$ problem is NP-hard even for~2 scenarios. But, if the number of scenarios is constant, then for all weight distributions 
 \textsc{Min-Owa}~$\mathcal{P}$  admits a fully polynomial time approximations scheme if only the corresponding exact problem can be solved in pseudopolynomial time. This is, however, only a theoretical result, because the FPTAS is typically exponential in the number of scenarios. If the number of scenarios is unbounded, then the problem becomes strongly NP-hard and two general approximation properties can be established.  If the weights are nonincreasing, then the problem admits an approximation algorithm with the worst case ratio equal to  $w_1K$, if only the deterministic problem is polynomially solvable. The largest approximation ratio equal to $K$ occurs for the maximum criterion and it becomes smaller when more uniform weight distributions are used. On the other hand, if the weights are nondecreasing, then \textsc{Min-owa}~$\mathcal{P}$ is not at all approximable for some basic network problems such as \textsc{Shortest Path}, \textsc{Minimum Spanning Tree}, \textsc{Minimum Assignment} and \textsc{Minimum s-t Cut}. This negative result remains true when OWA is median. All the new and known results for the \textsc{Min-Owa Shortest Path} problem are summarized in Table~\ref{tabsum}. A similar table can be shown for other particular problems \textsc{Min-Owa}~$\mathcal{P}$.
\begin{table}
\centering
\caption{Summary of the known and new results for the \textsc{Min-Owa Shortest Path} problem.} \label{tabsum}
\footnotesize
\begin{tabular}{l|l|l|l}  
\hline 
		Problem & $K=2$ &  $K\geq 3$ constant & $K$ unbounded \\ \hline
\textsc{Min-Owa}~$\mathcal{P}$ & equivalent to & NP-hard & strongly NP-hard  \\
                               & \textsc{Min-Hurwicz}~$\mathcal{P}$ &FPTAS  &  appr. within $w_1K$ if the weights\\ 
															 &                                    &        &  are nonincreasing\\                  
															 &                                    &        & not at all appr. if the weights \\
															 &                                    &        & are nondecreasing                 \\\hline \hline
\textsc{Min-Max}~$\mathcal{P}$ & NP-hard & NP-hard& strongly NP-hard   \\
                               & FPTAS   & FPTAS  &  appr. within $K$ \\
															 &          &      & not appr. within  \\
															  &         &      & $O(\log^{1-\epsilon}K)$, $\epsilon>0$ \\ \hline
\textsc{Min-Min}~$\mathcal{P}$ & poly. solvable  & poly. solvable  & poly. solvable  \\ \hline
\textsc{Min-Average}~$\mathcal{P}$ & poly. solvable  & poly. solvable  & poly. solvable  \\ \hline
\textsc{Min-Hurwicz}~$\mathcal{P}$ & poly. solvable  if $\alpha\in [0,1/2)$ & NP-hard if $\alpha\in (0,1]$ & strongly NP-hard if $\alpha \in (0,1]$  \\
																	 & NP-hard if $\alpha \in (1/2,1]$ & FPTAS & appr. within \\
																	 & FPTAS if $\alpha \in (1/2,1]$ &         & $\alpha K+ (1-\alpha)(K-2)$ if $\alpha\in [1/2,1]$ \\
																	 &                                &         & $K/\alpha$ if $\alpha \in (0,1/2)$ \\
																		 &          &      & not appr. within  \\ 
															  &         &      & $O(\log^{1-\epsilon}K)$, $\epsilon>0$ \\ \hline
\textsc{Min-Quant}$(k)$~$\mathcal{P}$ & poly. solvable if $k=2$ & poly. solvable if $k=K$ & strongly NP-hard for any $k\in [K-1]$ \\
																		&  NP-hard if $k=1$ & NP-hard for any  & approx. within $K$ when $k$ is constant  \\
																		&  FPTAS     &  constant $k\in[K-1]$     &        not at all appr. if $k=\lfloor K/2 \rfloor +1$ \\
&  & FPTAS 																		
																	
																		 \end{tabular}
\end{table}

Our goal has been to provide general properties of  \textsc{Min-Owa}~$\mathcal{P}$, 
which follow only from the type of  the weight distribution in the OWA operator. 
We have not taken  into account a particular structure of an
underling  deterministic problem $\mathcal{P}$. Thus,
the results obtained may be 
 additionally refined 
 if some properties of $\mathcal{P}$ are taken into account.

%***********Acknowledgements***********************************
\subsubsection*{Acknowledgements}
This work was 
partially supported by
 the National Center for Science (Narodowe Centrum Nauki), grant  2013/09/B/ST6/01525.

%\bibliographystyle{abbrv} % outcomment this and next line in Case 1
%\bibliography{owa} % if more than one, comma separated

%%%%%%%%%%%%%%%%%
\end{document}